\documentclass[reqno,onecolumn,letterpaper,11pt]{amsart}

\usepackage{amsmath,amssymb,amsthm,fullpage,calc,enumerate}
\usepackage{graphicx,subfigure,psfrag,multicol}
\usepackage{algorithm,algorithmic,array,tikz}
\usepackage{enumitem}
\usepackage[left=1in,right=1in,top=1in,bottom=1in]{geometry}
\usetikzlibrary{decorations.markings}
\usepackage{enumitem,linegoal}
\usepackage{xcolor}
\usepackage{mathtools}
\usepackage{xparse}
\usepackage{bbm}
\usepackage{hyperref}

\allowdisplaybreaks[1] \numberwithin{equation}{section}
\newtheorem{theorem}{Theorem}[section]

\newtheorem{lemma}[theorem]{Lemma}

\theoremstyle{definition}

\DeclareDocumentCommand\RR{}{\mathbb{R}}

\DeclareDocumentCommand\QQ{}{\mathbb{Q}}

\newcommand{\Tn}{\mathcal{T}_n}
\newcommand{\Tnm}{\mathcal{T}_{n-1}}
\newcommand{\mm}{\mathbb{M}}

\DeclareDocumentCommand\onevec{o}{\IfNoValueTF{#1}{\mathbbm{1}}{\mathbbm{1}_{#1}}}
\DeclareDocumentCommand\orderO{m}{\mathcal{O}(#1)}
\DeclareDocumentCommand\cplxNP{}{\mathsf{NP}}
\DeclareDocumentCommand\cplxCoNP{}{\mathsf{coNP}}

\DeclareDocumentCommand\transpose{m}{#1^{\intercal}}

\tikzstyle{vrtx} =[circle,fill=black,draw=black,inner sep=0pt,minimum size=3pt]
\tikzstyle{labvrtx}=[circle,fill=white,draw=black,inner sep=0.02cm, minimum size=0.5cm]

{\begin{list}{}{%
   \settowidth{\labelwidth}{\textbf{#1:}}%
   \setlength{\parsep}{0.0cm}%
   \setlength{\itemsep}{0.10cm}%
   \setlength{\leftmargin}{\labelwidth+\labelsep}}}%
{\end{list}}

{\begin{list}{}{%
   \settowidth{\labelwidth}{\textsf{#1)}}%
   \setlength{\parsep}{0.0cm}%
   \setlength{\itemsep}{0.10cm}%
   \setlength{\leftmargin}{\labelwidth+\labelsep}}}%
{\end{list}}
\title{A characterization of Linearizable instances of the Quadratic Traveling Salesman Problem}

\author[Abraham P. Punnen, Brad D. Woods]{ Abraham P. Punnen, Matthias Walter, and  Brad D. Woods}
\address{Abraham P. Punnen, Department of Mathematics, Simon Fraser University Surrey, 250-13450 102nd AV, Surrey, British Columbia, V3T 0A3, Canada, Email: {\tt apunnen@sfu.ca}}
\address{Matthias Walter, RWTH Aachen University, Kackertstra{\ss}e 52072 Aachen, Germany {\tt walter@or.rwth-aachen.de}}
\address{Brad Woods, Department of Mathematics, Simon Fraser University Surrey, 250-13450 102nd AV, Surrey, British Columbia, V3T 0A3, Canada. {\tt bdw2@sfu.ca}}

\begin{document}

\date{\today}

\begin{abstract}
  We consider the linearization problem associated with the quadratic traveling salesman problem (QTSP).
  Necessary and sufficient conditions are given for a cost matrix $Q$ for the QTSP to be linearizable.
  It is shown that these conditions can be verified in $\orderO{n^5}$ time.
  Some simpler sufficient conditions for linearization are also given along with related open problems.
\end{abstract}

\maketitle

\section{Introduction}

The Traveling Salesman Problem (TSP) is the problem of finding a least-cost Hamiltonian cycle in an edge-weighted graph.
It is one of the most widely studied hard combinatorial optimization problems.
For details on the TSP, we refer to the well-known books~\cite{ApplegateBCC11, GutinP06, LawlerLRS85, Reinelt94}.
For clarity of discussion, we will refer to this problem as the \emph{linear TSP}.

Let $G = (V,E)$ be a graph (directed or undirected) on the vertex set $ V = \{1, \ldots ,n\}$.
For each pair $(e,f)$ of edges in $E$, a cost $q_{ef}$ is prescribed.
Let $\Tn$ be the set of all Hamiltonian cycles (tours) in $G$.
The quadratic cost $Q[\tau]$ of a tour $\tau \in \Tn$ is given by $ Q[\tau] := \sum \limits_{(e,f) \in \tau \times \tau} q_{ef}$.
Then the \emph{quadratic traveling salesman problem} (QTSP) is to find a tour $\tau \in \Tn$ such that $Q[\tau]$ is minimum.

The problem QTSP received only limited attention in the literature.
A special case of the \emph{adjacent quadratic TSP} in which $q_{ef}$ is assumed to be zero if edges $e$ and $f$ are not adjacent, is frequently also called quadratic TSP.
This problem has been studied by various authors recently~\cite{Fischer14, Fischer16, FischerFJKMG14, FischerH13, JagerM08, RostamiMBG16, AichholzerFFMPPS17}.
While QTSP is known to be strongly $\cplxNP$-hard for Halin graphs, the adjacent quadratic TSP can be solved in $\orderO{n}$ time on this class of graphs~\cite{WoodsP17Halin}.
QTSP defined over the set of pyramidal tours is also known to be strongly $\cplxNP$-hard~\cite{WoodsP17PQ}.
Other special cases of QTSP are studied in~\cite{WoodsP17} along with various complexity results.
The $k$-neighbour TSP~\cite{WoodsPS17}, Angular metric TSP~\cite{AggarwalKMS97}, Dubins vehicle TSP~\cite{SavlaFB08} are also related to QTSP.

For the rest of this paper, unless otherwise stated, we assume that $G$ is a complete directed graph.
Let $Q = (q_{ef})$ be the matrix of size $n(n-1)\times n(n-1)$ consisting of the cost of the edge pairs.
Note that $(e,e)$ is a permitted pair.
An instance of QTSP is completely specified by $Q$.
The matrix $Q$ is said to be \emph{tour-linearizable} (or simply linearizable) if there exists a matrix $C = (c_{ij})$ of size $n \times n$ such that $Q[\tau] = C(\tau)$ for all $\tau \in \Tn$, where $C(\tau) := \sum_{(i,j)\in \tau} c_{ij}$.
Such a cost matrix $C$ is called a \emph{linearization} of $Q$ for the QTSP.

The \emph{QTSP linearization problem} can be stated as follows:
``Given an instance of QTSP with cost matrix $Q$, verify if it is linearizable and if yes, compute a linearization $C$ of $Q$.''

It may be noted that since (the decision versions of) the QTSP and the TSP are both $\cplxNP$-complete problems, there exists a polynomial-time reduction from one problem to the other.
We are not attempting to find a reduction from QTSP to TSP for some special cases, although such a result is a byproduct.
The question of linearization is more involved and requires the development of structural properties of linearizable matrices.
Unlike QTSP or TSP, there are no immediate certificates for linearizability or non-linearizability of matrices.
Hence, even containment of the linearization problem in $\cplxNP$ or $\cplxCoNP$ is challenging.

Linearization problems have been studied in the literature for various quadratic combinatorial optimization problems~\cite{AdamsW14, cCelaDW16, CusticSPB17, CusticP18, KabadiP11, LendlP17, PunnenK13, PunnenP18, HuS18, HuS18a} and polynomial-time algorithms are available in all these cases.
Unfortunately, the simple characterization established for quadratic spanning trees~\cite{CusticP18} and for various bilinear type problems~\cite{cCelaDW16, CusticSPB17, LendlP17} do not seem to extend to the case of QTSP.
The characterization of linearizable instances of the quadratic assignment problem studied in~\cite{KabadiP11} have some similarities to the case of QTSP that we study here.
However, to solve the linearization problem for QTSP one needs to overcome more challenges, which makes it an interesting and relevant problem to investigate.
A family of linearizable cost matrices can be used to obtain strong continuous relaxations of quadratic 0-1 integer programming problems, as demonstrated in~\cite{HuS18a, PunnenP18}.
Strong reformulations of quadratic 0-1 integer programming problems are reported in~\cite{BillionnetEP09, PornNSW17} exploiting linear equality constraints.
Moreover, strong reformulations of quadratic 0-1 integer programming using quadratic equality constraints are reported in~\cite{GalliL14}.
As observed in~\cite{PunnenP18}, linearizable cost matrices immediately give redundant quadratic equality constraints and, exploiting the results of~\cite{GalliL14}, strong reformulations of quadratic combinatorial optimization problems can be obtained.
Further, since linearizable cost matrices for any quadratic combinatorial optimization problem form a linear subspace of the space of all $n(n-1) \times n(n-1)$ matrices, any basis of this subspace can be used to obtain strong reformulations~\cite{PunnenP18}.

With this motivation, we study the QTSP linearization problem in this paper.
Necessary and sufficient conditions are presented for a quadratic cost matrix $Q$ associated with QTSP being linearizable.
We show that these conditions can be verified by an $\orderO{n^5}$-time algorithm that also constructs a linearization of $Q$ whenever one exists.
However, it may be noted that the input size of a general QTSP is $\orderO{n^4}$ and hence the complexity of our algorithms is almost linear in the input size in the general case.
Our characterization extends directly to the case of complete undirected graphs.
Finally, some easily verifiable sufficient conditions for linearizability are also given along with some open problems.

The terminology ``linearization'' has also been used in a different context where a quadratic integer program is written as a linear integer program by adding new variables and constraints~\cite{AdamsFG04,Glover75}.
We want to emphasize that the concept of linearization used in this paper is different.

We conclude this section by presenting the most-relevant notation used in the paper.
Let $N := \{ 1,2,\ldots,n\}$ and $N_i := N \setminus \{ i\}$ for $i \in N$.
All matrices are represented using capital letters (sometimes with superscripts, hats, overbars, etc.) and all elements of the matrix are represented by the corresponding subscripted lower-case letters (along with the corresponding superscripts, hats, overbars, etc.), where the subscripts denoting row and column indices.
When the rows and columns are indexed by elements of $N \times N$, the $((i,j),(k,l))$'th element of the matrix $Q$ is represented by $q_{ijkl}$, of $Q^R$ by $q^R_{ijkl}$, etc.
Vectors in $\RR^n$ are represented by bold lower-case letters (sometimes with superscripts, hats, overbars, etc.).
The $i$'th component of vector $\mathbf{a}$ is $a_i$, of vector $\bar{\mathbf{b}}$, is $\bar{b}_i$, etc.
Rows and columns of all matrices of size $n \times n$ are indexed by $N$, whereas rows and columns of all $n(n-1) \times n(n-1)$ matrices are indexed by edges of $G$ (i.e., by ordered pairs $(i,j) \in N \times N$, $i \neq j$).
The transpose of a matrix $Q$ is denoted by $\transpose{Q}$.
The vector space of all real-valued $n \times n$ matrices with standard matrix addition and scalar multiplication is denoted by $\mathbb{M}^n$.
Accordingly, $\mathbb{M}^{n(n-1)}$ is the vector space of all $n(n-1) \times n(n-1)$ matrices.

\section{The QTSP linearization problem}

Let us first introduce some general properties and definitions used in the paper.
A matrix $C \in\mm^n$ satisfies the \emph{constant value property} (CVP) if there exists a constant $K$ such that $C(\tau) = K$ for all $\tau \in \Tn$.
A matrix $C$ is said to be a \emph{weak sum matrix} if there exist vectors $\mathbf{a},\mathbf{b}\in \mathbb{R}^n$ such that $c_{ij} = a_i + b_j$ for $i,j \in N$ with $i \neq j$.

\begin{lemma}[\cite{Berenguer79, Gabovich76, KabadiP03, KabadiP06}]
  A matrix $C \in \mm^n$ satisfies the CVP if and only if it is a weak sum matrix (i.e., $c_{ij} = a_i + b_j$ for all $i,j \in N$ with $i \neq j$).
  In this case, the constant value of the tours is given by $\sum_{i\in N}(a_i + b_i)$.
\end{lemma}

Two matrices $Q^1$ and $Q^2 \in \mm^{n(n-1)}$ are \emph{equivalent modulo a linear matrix} if there exists a matrix $H \in \mm^n$ such that $Q^1[\tau] = Q^2[\tau] + H(\tau)$ for all $\tau \in \Tn$.
Using the fact that the value of a tour depends on the sum $q_{ijk\ell} + q_{k\ell ij}$ and not the individual values $q_{ijk\ell}$ and $q_{k\ell ij}$, we obtain the following lemma.
\begin{lemma}
  \label{LemmaInvariantTransformations}
  Let $Q \in \mm^{n(n-1)}$.
  Then the following statements are equivalent.
  \begin{enumerate}
  \item
    $Q$ is linearizable.
  \item
    $\transpose{Q}$ is linearizable.
  \item
    $\frac{1}{2}\left(Q+\transpose{Q} \right)$ is linearizable.
  \item
    $Q+D+S$ is linearizable, where $D$ is any diagonal matrix and $S$ is any skew-symmetric matrix of the same dimension as $Q$.
  \item
    Any matrix equivalent to $Q$ modulo a linear matrix is linearizable.
  \end{enumerate}
\end{lemma}

Recall that a matrix $C \in \mm^{n}$ is a \emph{linearization} of $Q \in \mm^{n(n-1)}$ if $Q[\tau] = C(\tau)$ for all $\tau \in \Tn$.
In this case we say that $Q$ is a \emph{quadratic form} of $C$ and that the matrix $Q$ is \emph{linearizable}.
Note that the sum of two linearizable matrices in $\mm^{n(n-1)}$ is linearizable and the scalar multiple of a linearizable matrix is linearizable.
Hence, the collection of all linearizable matrices in $\mm^{n(n-1)}$ forms a linear subspace of $\mm^{n(n-1)}$.

The coefficients $q_{iji\ell}$ with $j \neq \ell$, the coefficients $q_{ijkj}$ with $i \neq k$ and the coefficients $q_{ijji}$ can never contribute a non-zero value to the objective function $Q[\tau]$ of a tour $\tau$.
These elements of the matrix $Q$ they play no role in the optimization problem, therefore, we assume subsequently that $q_{ijk\ell} = 0$ for all indices $i,j,k$ and $\ell$ that satisfy $i = k$ and $j \neq \ell$ or $i \neq k$ and $j = \ell$ or $i = \ell$ and $j = k$.

To characterize the linearizability of $Q$, working with the original matrix $Q$ seems tedious because of several reasons.
For example, as Lemma~\ref{LemmaInvariantTransformations} indicates, we can add any skew-symmetric matrix to $Q$, maintaining linearizability.
It is important to restrict $Q$ in an appropriate way to bring some kind of uniqueness while the generality is maintained.
For this purpose, we say that a matrix $Q \in \mm^{n(n-1)}$ is in \emph{quadratic reduced form} if all of the following conditions are satisfied:

\begin{enumerate}
\item
  $Q$ is symmetric.
\item
  The diagonal elements of $Q$ are zeros.
\item
  All elements of $Q$ with rows and columns indexed by $\{ (n,p),(p,n) : p\in N_n\}$ are zeros.
\end{enumerate}

Let us now prove a useful decomposition theorem.

\begin{theorem}
  \label{TheoremReducedFormDecomposition}
  Any matrix $Q \in \mm^{n(n-1)}$ can be decomposed in $\orderO{n^4}$ time into a pair $(Q^R, L)$ of matrices such that $Q^R \in \mm^{n(n-1)}$ is in quadratic reduced form, $L \in \mm^n$ and $Q[\tau] = Q^R[\tau] + L(\tau)$ for all $\tau \in \Tn$.
\end{theorem}

\begin{proof}
  Let $\tilde{Q}$ and $\tilde{L}$ be defined as
  \[
    \tilde{q}_{ijkl} =
    \begin{cases}
      q_{ijk\ell} - q_{ijkn} - q_{ijn\ell} - q_{ink\ell} + q_{inkn} + q_{inn\ell} - q_{njk\ell} + q_{njkn} +q_{njn\ell} & \text{ if } n \notin \{i,j,k,\ell\} \\
      0 & \text{ otherwise}
    \end{cases}
  \]
  and
  \[
    \tilde{l}_{ij} =
    \begin{cases}
      \sum \limits_{k=1}^{n-1}(q_{ijkn} + q_{ijnk} + q_{knij} - q_{knin} - q_{knnj} + q_{nkij} - q_{nkin} - q_{nknj}) & \text{ if } i \neq j \\
      0 & \text{ otherwise.}
    \end{cases}
  \]
  It can be verified that all entries of rows and columns of $\tilde{Q}$ indexed by $(n,p)$ or $(p,n)$ for $p \in N_n$ are zeros.
  Further, with some algebra, it can be verified that $\tilde{Q}$ is equivalent to $Q$ modulo the linear matrix $\tilde{L}$.
  Let $D \in \mm^{n(n-1)}$ be the diagonal matrix with $d_{ijij} = \tilde{q}_{ijij}$.
  Define $Q^R= \frac{1}{2}(\tilde{Q}+\transpose{\tilde{Q}}) - D$ and choose $L \in \mm^n$ such that $l_{ij} = \tilde{l}_{ij} + d_{ijij}$.
  From the discussion above and Lemma~\ref{LemmaInvariantTransformations} we have $Q[\tau] = Q^R[\tau] + L(\tau)$ for all $\tau \in \Tn $ and that $Q^R$ is in quadratic reduced form.
  Finally, it is easy to see that only $\orderO{n^4}$ arithmetic operations are required to compute $Q^R$ and $L$.
\end{proof}

The matrix $Q^R$ constructed in the proof of Theorem~\ref{TheoremReducedFormDecomposition} is referred to as the \emph{quadratic reduced form of $Q$}.
The decomposition $(Q^R,L)$ of $Q$ is referred to as \emph{quadratic reduced form decomposition} or simply \emph{QRF decomposition}.
The proof of Theorem~\ref{TheoremReducedFormDecomposition} provides one way of constructing a QRF decomposition of $Q$.
It may be noted that equivalent transformations discussed in~\cite{Burkard07,KabadiP11} in the context of the quadratic assignment problem can be modified appropriately to get another method for constructing a QRF decomposition.

\begin{theorem}
  Let $Q \in \mm^{n(n-1)}$ and $(Q^R,L)$ be a QRF decomposition of $Q$.
  Then $Q$ is linearizable if and only if $Q^R$ is linearizable.
  Further, if $P \in \mm^n$ is a linearization of $Q^R$, then $P + L$ is a linearization of $Q$.
\end{theorem}

\begin{proof}
  Since $(Q^R,L)$ is a QRF decomposition of $Q$, $Q[\tau] = Q^R[\tau] + L(\tau)$ holds for all $\tau \in \Tn$.
  By Lemma~\ref{LemmaInvariantTransformations}, $Q$ is linearizable if and only if $Q^R$ linearizable.
  If $P$ is a linearization of $Q^R$, it follows that $Q[\tau] = (P + L)(\tau)$ for all $\tau \in \Tn$.
\end{proof}

Thus, to study the linearization problem of $Q$, it is sufficient to study the linearization problem of the quadratic reduced form $Q^R$ of $Q$.
The lemma below is a variation of the corresponding result proved in~\cite{KabadiP11} for the quadratic assignment problem.
Since it is used in our main theorem, a complete proof within the context of QTSP is given below.
\begin{lemma}
  \label{LemmaLinearizationWithZeros}
  Suppose $Q^R \in \mm^{n(n-1)}$ is linearizable and is in quadratic reduced form.
  Then there exists a linearization $C$ of $Q^R$ such that $c_{in} = 0$ and $c_{ni} = 0$ for all $i \in N_n$.
\end{lemma}

\begin{proof}
  Suppose $C'$ is some linearization of $Q^R$.
  Let $\alpha := \sum_{i\in N_n}c'_{in} / (n-2)$ and $\beta := \sum_{i\in N_n}c'_{ni}/(n-2)$.
  Define $a_i := -c'_{in} + \beta$, $b_i := -c'_{ni} + \alpha$ for all $i \in N_n$ as well as $a_n := -\alpha$ and $b_n := -\beta$.
  Moreover, let $c_{ij} := c'_{ij} + a_i + b_j$ for all $i,j \in N$.
  Then $\sum_{i\in N}(a_i+b_i) = 0$ and hence $C(\tau) = C^{\prime}(\tau)$ for all $\tau \in \Tn$.
  Now it can be verified that $c_{in} = c_{ni} = 0$ for all $i \in N_n$, which concludes the proof.
\end{proof}

Let $\bar{Q}$ be the principal submatrix obtained from $Q^R$ by deleting rows and columns indexed by the elements of the set $\{ (n,p),(p,n) : p\in N_n\}$.
Note that $\bar{Q} \in \mm^{(n-1)(n-2)}$.
For any $i,j \in N_n$ with $i \neq j$, define $Z^{ij} \in \mm^{n-1}$ such that
\begin{equation*}
  z^{ij}_{uv} =
  \begin{cases}
     0 & \text{ if } (i,j) = (u,v), \\
     \bar{q}_{ijuv}  & \text{if } (i,j) \neq (u,v).
  \end{cases}
\end{equation*}
Our next theorem characterizes the linearizability of a matrix in quadratic reduced form.
Note that, for $n \leq 3$, any quadratic cost matrix is linearizable.
So, we restrict our attention to the case when $n \geq 4$.

\begin{theorem}
  \label{TheoremCharacterization}
  Let $Q^R \in \mm^{n(n-1)}$ be in quadratic reduced form and $n \geq 4$.
  Then $Q^R$ is linearizable if and only if there exists a matrix $F \in \mm^{n-1}$ (with elements $f_{ij}$) such that
  \begin{enumerate}
  \item
    \label{TheoremCharacterizationFirst}
    $Z^{ij}(\tau) - Z^{k\ell}(\tau) = \frac{1}{2} \left(f_{ij} - f_{kl}\right )$ holds for every $\tau \in \Tnm$ and for all $(i,j), (k,\ell) \in \tau $, and
  \item
    \label{TheoremCharacterizationSecond}
    $\frac{n-2}{n-3}F$ is a linearization of $\bar{Q}\in \mm^{(n-1)(n-2)}$.
  \end{enumerate}
\end{theorem}

\begin{proof}
Suppose $Q^R$ is linearizable.
Let $C$ be a linearization of $Q^R$ satisfying the condition of Lemma~\ref{LemmaLinearizationWithZeros} and let $C'$ be the submatrix obtained from $C$ by deleting its $n$'th row and column.
Define $F$ via $f_{ij} := c_{ij}$ for all $i,j \in N_n$ with $i\neq j$.
We will show that $F$ satisfies the conditions of the theorem.
Consider any $\tau \in \Tnm$.
For any $(i,j) \in \tau$, we define the tour $\tau^{ij} \in \Tn$ obtained by deleting arc $(i,j)$ from $\tau$ and introducing arcs $(i,n)$ and $(n,j)$ (See Figure~\ref{FigureTours}).

\begin{figure}[H]
  \centering
  \begin{tikzpicture}[>=stealth, decoration={markings, mark=at position 0.5 with {\arrow{>}}}]
	\node [vrtx] (v1) at (-3, 1.5){};
	\node [vrtx] (v2) at (-1.8273, 0.9352) {};
    \node [vrtx] (v3) at (-1.5376, -0.3338) {};
    \node [vrtx, label=below right:{\footnotesize$j$}] (v4) at (-2.3492, -1.3515) {};
    \node [vrtx, label=below left:{\footnotesize$i$}] (v5) at (-3.6508, -1.3515) {};
    \node [vrtx] (v6) at (-4.4624, -0.3338) {};
    \node [vrtx] (v7) at (-4.1727, 0.9352) {};
	\draw[postaction={decorate}] (v5) -- (v4);
	\draw[postaction={decorate}] (v6) -- (v5);
	\draw[postaction={decorate}] (v7) -- (v6);
	\draw[postaction={decorate}] (v1) -- (v7);
	\draw[postaction={decorate}] (v2) -- (v1);
	\draw[postaction={decorate}] (v3) -- (v2);
	\draw[postaction={decorate}] (v4) -- (v3);

	\node [vrtx] (u1) at (3, 1.5){};
	\node [vrtx] (u2) at (4.1727, 0.9352) {};
    \node [vrtx] (u3) at (4.4624, -0.3338) {};
    \node [vrtx, label=below right:{\footnotesize$j$}] (u4) at (3.6508, -1.3515) {};
    \node [vrtx, label=below left:{\footnotesize$i$}] (u5) at (2.3492, -1.3515) {};
    \node [vrtx] (u6) at (1.5376, -0.3338) {};
    \node [vrtx] (u7) at (1.8273, 0.9352) {};
    \node [vrtx, label=above:{\footnotesize$n$}] (n) at (3, -0.5) {};
	\draw[postaction={decorate}] (u5) -- (n);
	\draw[postaction={decorate}] (u6) -- (u5);
	\draw[postaction={decorate}] (u7) -- (u6);
	\draw[postaction={decorate}] (u1) -- (u7);
	\draw[postaction={decorate}] (u2) -- (u1);
	\draw[postaction={decorate}] (u3) -- (u2);
	\draw[postaction={decorate}] (u4) -- (u3);
	\draw[postaction={decorate}] (n)-- (u4);
	\draw[dashed, postaction={decorate}] (u5) -- (u4);
  \end{tikzpicture}
  \caption{Tours $\tau \in \Tnm$ (shown on the left) and $\tau^{ij}\in \Tn$ (shown on the right) with $n = 8$.}
  \label{FigureTours}
\end{figure}

Note that $Q^R[\tau^{ij}] = \bar{Q}[\tau] - 2Z^{ij}(\tau)$ and $C(\tau^{ij}) = C^{\prime}(\tau) - c_{ij}$.
But $Q^R[\tau_{ij}] = C(\tau_{ij})$, and we obtain
\begin{equation}
  \label{TheoremCharacterizationZ1}
  Z^{ij}(\tau) = \frac{1}{2} \left( \bar{Q}[\tau] - C^{\prime}(\tau) +c_{ij}\right).
\end{equation}
Similarly, for $(k,\ell) \in \tau$,
\begin{equation}
  \label{TheoremCharacterizationZ2}
  Z^{k\ell}(\tau) = \frac{1}{2} \left( \bar{Q}[\tau] - C^{\prime}(\tau) +c_{k\ell}\right).
\end{equation}
Subtracting~\eqref{TheoremCharacterizationZ2} from~\eqref{TheoremCharacterizationZ1} yields
\[
  Z^{ij}(\tau) - Z^{k\ell}(\tau) = \frac{1}{2} \left(f_{ij} - f_{kl}\right),
\]
i.e., condition~\eqref{TheoremCharacterizationFirst} of the theorem is satisfied.
Now, adding~\eqref{TheoremCharacterizationZ1} for all $(i,j)\in \tau$, we obtain
\begin{align*}
  \bar{Q}[\tau] & = \sum_{(i,j)\in \tau}Z^{ij}(\tau) =
  \sum_{(i,j)\in \tau} \frac{1}{2} \left( \bar{Q}[\tau] - C^{\prime}(\tau) +c_{ij}\right)
  = \frac{1}{2} \left[(n-1)\bar{Q}[\tau] - (n-2)C^{\prime}(\tau)\right],
\end{align*}
which implies $\bar{Q}[\tau] = \frac{n-2}{n-3}C^{\prime}(\tau)$.
Thus, $\bar{Q}$ is tour-linearizable and $\frac{n-2}{n-3}C^{\prime}(\tau)$ is a linearization.

\medskip

Conversely, suppose $Q^R$ satisfies conditions~\eqref{TheoremCharacterizationFirst} and~\eqref{TheoremCharacterizationSecond} of the theorem.
We will show that $Q^R$ is tour-linearizable.
To this end, define the matrix $C \in \mm^n$ via
\begin{equation}
  \label{TheoremCharacterizationC}
  c_{ij} :=
  \begin{cases}
  f_{ij} & \text{ if } i,j \in N_n, \\
  0 & \text{ otherwise}.
  \end{cases}
\end{equation}
Let $C^{\prime}$ be the principal submatrix obtained from $C$ by deleting its $n$'th row and column. From condition~\eqref{TheoremCharacterizationFirst} of the theorem, for every tour $\tau \in \Tnm$ and every pair $(i,j),(k,\ell) \in \tau$ of edges in the tour,
\begin{equation}
  \label{TheoremCharacterizationPhi}
  2Z^{ij}(\tau) = c_{ij} + \phi(\tau)
\end{equation}
holds, where $\phi(\tau) = 2Z^{k\ell}(\tau) - c_{k\ell}$.
From condition~\eqref{TheoremCharacterizationSecond}, for every $\tau \in \Tnm$, $\bar{Q}[\tau] = \frac{n-2}{n-3}C^{\prime}(\tau)$, which implies
\begin{align}
  0 &= (n-1)\bar{Q}[\tau] - 2\bar{Q}[\tau] - (n-2)C^{\prime}(\tau) \nonumber \\
    &= (n-1)\bar{Q}[\tau] - \sum_{(i,j)\in \tau}2Z^{ij}(\tau) - (n-2)C^{\prime}(\tau). \label{TheoremCharacterizationInsertZ}
\end{align}
Substituting equation~\eqref{TheoremCharacterizationPhi} in~\eqref{TheoremCharacterizationInsertZ} yields
\begin{align}
  0 &= (n-1)\bar{Q}[\tau] - \sum_{(i,j)\in \tau}\left[c_{ij}+\phi(\tau) \right]- (n-2)C^{\prime}(\tau) \nonumber \\
    &= (n-1)\bar{Q}[\tau] - C^{\prime}(\tau) - (n-1)\phi(\tau)- (n-2)C^{\prime}(\tau), \nonumber
\end{align}
which implies
\begin{equation}
  \phi(\tau) = \bar{Q}[\tau] - C^{\prime}(\tau) \label{TheoremCharacterizationPhiNew}
\end{equation}
Substituting equation~\eqref{TheoremCharacterizationPhiNew} in~\eqref{TheoremCharacterizationPhi}, we have, for every $(i,j) \in \tau$,
\begin{equation}
  2Z^{ij}(\tau) = c_{ij} + \bar{Q}[\tau] - C^{\prime}(\tau). \label{TheoremCharacterizationZNew}
\end{equation}
Consider any tour $\tau^{ij} \in \Tn$, where $i$ and $j$ are the predecessor and successor of $n$ in $\tau^{ij}$, respectively.
Let $\tau^0$ be the tour in $\Tnm$ obtained from $\tau^{ij}$ by shortcutting the path $i \rightarrow n \rightarrow j$ using arc $(i,j)$.
Then,
\begin{align*}
  Q^R[\tau^{ij}]  &= \bar{Q}[\tau^0] - 2Z^{ij}(\tau^0) \\
                  &= \bar{Q}[\tau^0] -\left[c_{ij} + \bar{Q}[\tau^0] -C^{\prime}(\tau^0)\right] \\
                  &= C^{\prime}(\tau^0) - c_{ij} = C(\tau^{ij}),
\end{align*}
establishing that $Q^R$ is linearizable with $C$ as a linearization.
\end{proof}

\begin{theorem}
  The QTSP linearization problem can be solved in $\orderO{n^5}$ time.
  Moreover, a linearization of $Q$ can be constructed in $\orderO{n^5}$ time whenever it exists.
\end{theorem}

\begin{proof}
  Given $Q$, its QRF decomposition $(Q^R,L)$ can be obtained in $\orderO{n^4}$ time using the construction from Theorem~\ref{TheoremReducedFormDecomposition}.
  We now test the linearizability characterization of $Q^R$ given in Theorem~\ref{TheoremCharacterization}.
  Let $\Delta := \{[(i,j),(k,\ell)] : i,j,k,\ell \in N_n, k\neq i, j\neq \ell, (i,j)\neq (j,i), i\neq j, k \neq \ell\}$.
  To verify condition~\eqref{TheoremCharacterizationFirst} of Theorem~\ref{TheoremCharacterization} and construct a linearization, if exists, we have to find $f_{ij}, (i,j) \in N_n, i \neq j$, such that for any $\tau \in \Tn,$  $Z^{ij}(\tau) - Z^{k\ell}(\tau) = \frac{1}{2} \left(f_{ij} - f_{kl}\right )$ for all $(i,j), (k,\ell) \in \tau$.
  This can in principle be achieved by first testing if the matrix $P^{ijk\ell} := Z^{ij}-Z^{kl}$ satisfies the CVP.
  If it does not satisfy the CVP for some $[(i,j),(k,\ell)] \in \Delta$, then condition~\eqref{TheoremCharacterizationFirst} of Theorem~\ref{TheoremCharacterization} is not satisfied.
  If the CVP is satisfied for all $[(i,j),(k,\ell)] \in \Delta$, let $\eta_{ijk\ell}$ be the constant value of tours obtained for the matrix $P^{ijk\ell}$.
  Now we need to solve the system
  \begin{equation}
    \label{EquationAlgorithmSystem}
    f_{ij} - f_{k\ell} = 2\eta_{ijk\ell} \text{ for all } [(i,j),(k,\ell)] \in \Delta.
  \end{equation}
  If the system is inconsistent, condition~\eqref{TheoremCharacterizationFirst} of Theorem~\ref{TheoremCharacterization} is not satisfied.
  Using the characterization of TSP with CVP~\cite{Berenguer79,Gabovich76,KabadiP03,KabadiP06}, each of the $\orderO{n^4}$-many $\eta_{ijk\ell}$-values can be obtained in $\orderO{n^2}$ time.

  Since this would lead to a running time of $\orderO{n^6}$ just to compute the right-hand side of the system, we will now improve this part of the algorithm.
  To this end, consider the graph $\tilde{H} = (E,\tilde{\Delta})$, where
  $\tilde{\Delta} := \{ \{(i,j),(k,\ell) \in E \times E : [(i,j),(k,\ell)] \in \Delta \}$,
  i.e., two edges $(i,j)$ and $(k,\ell)$ are adjacent in $\tilde{H}$ if and only if both arcs appear in at least one tour.
  It is easy to check that $\tilde{H}$ is connected and hence contains a spanning tree $\tilde{T}$.
  Moreover, by exploring the nodes from some arbitrary root node, one finds a sequence of $|E|-1$ rows of~\eqref{EquationAlgorithmSystem},
  each of which contains an $f$-variable that was not present in the previous rows.
  Since the first row contains two $f$-variables, this submatrix has rank at least $|E|-1$.

  We can find in $\orderO{n^2}$ time a spanning tree of $\tilde{H}$ and compute the corresponding $(|E|-1)$-many $\eta$-values in $\orderO{n^2 \cdot |E|} = \orderO{n^4}$ time.
  Such $\eta$-values may not exist, and in this case condition~\ref{TheoremCharacterizationFirst} of Theorem~\ref{TheoremCharacterization} is not satisfied.
  If these $\eta$-values exist, we proceed as follows.
  Since every row of this system contains exactly two variables, it can be solved in time $\orderO{n^2}$ using the algorithm of Aspvall and Shiloach~\cite{AspvallS80}.
  Note that due to the selection of the rows the system will always be consistent, that is, we will compute a candidate solution $\tilde{F}$.
  Moreover, the system is under-determined, since $\tilde{F} + \lambda \onevec$ is also a solution for every $\lambda \in \QQ$, where $\onevec$ denotes the matrix having all entries equal to $1$.
  This shows that system~\eqref{EquationAlgorithmSystem} has in fact rank $|E|-1$.
  Notice also that two rows of~\eqref{EquationAlgorithmSystem} corresponding to $[(i,j),(k,\ell)] \in \Delta$ and to $[(k,\ell),(i,j)]$ are negatives of each other.
  Hence, it does not matter which one we actually include in the subsystem.

  We claim that our candidate solution $\tilde{F}$ satisfies all (remaining) rows of~\eqref{EquationAlgorithmSystem}.
  To this end, consider a pair $[(i,j),(k,\ell)] \in \Delta$ of arcs.
  Let $(i_t,j_t)$ for $t = 0,1,\dotsc,m$ be the (ordered) nodes in the unique $(k,\ell)$-$(i,j)$-path in the spanning tree $\tilde{T}$.
  In particular, $(i_0,j_0) = (k,\ell)$ and $(i_m,j_m) = (i,j)$.
  By construction, $\tilde{F}$ satisfies
  \begin{align*}
    \tilde{f}_{i_tj_t} - \tilde{f}_{i_{t-1}j_{t-1}} &= 2\eta_{i_tj_ti_{t-1}j_{t-1}}
  \end{align*}
  for $t=1,2,\dotsc,m$. The sum of the equations yields
  \begin{align*}
    \tilde{f}_{ij} - \tilde{f}_{k\ell} = \tilde{f}_{i_mj_m} - \tilde{f}_{i_0j_0} &= 2 \sum_{t=1}^m \eta_{i_tj_ti_{t-1}j_{t-1}}.
  \end{align*}
  Similarly, for all tours $\tau \in \mathcal{T}_{n-1}$ and for $t = 1,2,\dotsc,m$, we have
  \begin{align*}
    Z^{i_tj_t}(\tau) - Z^{i_{t-1}j_{t-1}}(\tau) &= \eta_{i_tj_ti_{t-1}j_{t-1}}.
  \end{align*}
  Again, the sum yields
  \begin{align*}
    Z^{ij}(\tau) - Z^{k\ell}(\tau) &= \sum_{t=1}^m \eta_{i_tj_ti_{t-1}j_{t-1}},
  \end{align*}
  which establishes that the CVP is satisfied for $P^{ijk\ell}$ and that $\eta_{ijk\ell} = \sum_{t=1}^m \eta_{i_tj_ti_{t-1}j_{t-1}}$.
  Hence, $\tilde{F}$ satisfies the row of~\eqref{EquationAlgorithmSystem} corresponding to $[(i,j),(k,\ell)]$, which shows that this row is redundant.

  The parameter $\lambda$ can be easily determined by enforcing $Q^R[\tau] = (\tilde{F} + \lambda \onevec)(\tau)$ for an arbitrary tour $\tau$.
  Since $\onevec(\tau) = n$, we can set $F := \tilde{F} + \lambda \onevec$ for $\lambda := \tfrac{1}{n}(Q^R[\tau] - \tilde{F}(\tau))$.

  To summarize, the overall complexity of verifying condition~\eqref{TheoremCharacterizationFirst} of Theorem~\ref{TheoremCharacterization} and computing a candidate solution $F$ is $\orderO{n^4}$.

  If condition~\eqref{TheoremCharacterizationFirst} is satisfied, then we need to verify condition~\eqref{TheoremCharacterizationSecond}.
  This is achieved by testing if $\bar{Q}$ is linearizable and if yes, by testing if $\frac{n-2}{n-3}F$ is one of the linearizations.
  If $\bar{Q}$ is not linearizable, then condition~\eqref{TheoremCharacterizationSecond} fails.
  Suppose $H$ is a linearization of $\bar{Q}$.
  Then, to verify if $\frac{n-2}{n-3}F$ is another linearization, it is enough to verify if the matrix $H - \frac{n-2}{n-3}F$ satisfies the CVP with the constant tour value equal to zero.
  Given $H$, this can be achieved in $\orderO{n^2}$ time.

  Thus, the problem of testing linearizability of $Q$ reduces to that of testing linearizability of $\bar{Q}$, which has the size parameter $n$ reduced by 1 along with an additive $\orderO{n^4}$ effort.
  In the base case $n = 3$, $Q^R$ is trivially linearizable since there only exist two opposite tours that do not share edges.
  Thus, if $g(n)$ is the complexity of testing linearizability of $Q$, we have $g(n) = g(n-1) + \orderO{n^4}$, establishing $g(n) = \orderO{n^5}$.

  To obtain the linearization, construct the matrix $C$ from $F$ according to~\eqref{TheoremCharacterizationC}.
  Then $C + L$ is the required linearization of $Q$, where $L$ is provided by the QRF decomposition $(Q^R, L)$ of $Q$.
\end{proof}

\section{Simple sufficient conditions and extensions}

The complexity of testing linearizability of $Q$, although almost linear, imposes limits on applicability of the result outside the theoretical realm.
In this section we present some sufficient conditions for linearizability that can be verified more easily.

The cost matrix $C\in \mm^n$ associated with a linear TSP on a complete directed graph on $n$ nodes satisfies the \emph{$(k,\ell)$-constant value property} ($(k,\ell)$-CVP) if all tours containing the edge $(k,\ell)$ have the same cost.
\begin{lemma}
  A cost matrix $C \in \mm^n$ satisfies the $(k,\ell)$-CVP if and only if there exist $a_i \in N_k$ and $b_i\in N_{\ell}$ such that $c_{ij} = a_i + b_j$ for all $i,j \in N$ with $i\neq j$, $i \neq k$ and $j \neq \ell$.
  Moreover, the constant value of the tours is given by $c_{k\ell} + \sum_{i\in N_k} a_i + \sum_{i \in N_{\ell}} b_i$.
\end{lemma}
The proof of this lemma can be obtained by making use of the characterization of matrices satisfying the CVP on complete digraphs.
For each row indexed by $(i,j)$ of $Q$, define $R^{ij}\in \mm^n$ by
\begin{equation}
  r^{ij}_{uv} :=
  \begin{cases}
    q_{ijuv} & \text{if } u \neq v \\
    0& \text{ otherwise}.
  \end{cases}
\end{equation}

\begin{lemma}
  If $R^{ij}$ satisfies the $(i,j)$-CVP for all $(i,j) \in N \times N$ with $i\neq j$, then $Q$ is linearizable.
  Moreover, if $L \in \mm^n$ is a linearization, then $l_{ij}$ is equal to the $(i,j)$-CVP value of tours.
\end{lemma}

\begin{proof}
  For any $\tau \in \Tn$, $Q[\tau] = \sum_{(i,j)\in \tau}R^{ij}(\tau) = \sum_{(i,j)\in\tau}l_{ij} = L(\tau)$.
\end{proof}
A corresponding result can be derived by considering columns of $Q$.
Another simple sufficient condition for linearizability is that $Q$ is a weak-sum matrix.
A more general version of this condition is provided below.
\begin{lemma}
  If there exist $A,B,D,H \in \mathbb{R}^{n\times n\times n}$ such that $q_{ijkl} = a_{ijk} + b_{ijl} + d_{ikl} + h_{jkl}$ for all $i,j,k,l$ with $i\neq j$ and $k\neq l$ then $Q$ is linearizable.
\end{lemma}
\begin{proof}
  Let $\alpha_{ij} := \sum \limits_{k=1}^n a_{ijk}$, $\beta_{ij} := \sum \limits_{k=1}^n b_{ijk}$, $\gamma_{ij} := \sum \limits_{k=1}^n d_{kij}$, $\delta_{ij} := \sum \limits_{k=1}^n h_{kij}$ and $c_{ij} := \alpha_{ij} + \beta_{ij} + \gamma_{ij} + \delta_{ij}$ for all $i,j$ with $i\neq j$.
  Also, let $x_{ij}$ be the binary decision variable which takes value $1$ if and only if the underlying tour contains arc $(i,j)$.
  Thus, the values of $x_{ij}$ completely define a tour.
  Now,
\begin{align*}
  Q[\tau]
  &= \sum \limits_{i=1}^n \sum \limits_{j=1,j\neq i}^n \sum \limits_{k=1}^n \sum \limits_{l=1,k\neq l}^n q_{ijkl}x_{ij}x_{kl} \\
  &= \sum \limits_{i=1}^n \sum \limits_{j=1,j\neq i}^n \sum \limits_{k=1}^n \sum \limits_{l=1,k\neq l}^n (a_{ijk}+b_{ijl}+d_{ikl}+h_{jkl})x_{ij}x_{kl} \\
  &= \sum \limits_{i=1}^n \sum \limits_{j=1,j\neq i}^n x_{ij} \sum \limits_{k=1}^n a_{ijk} \sum \limits_{l=1,k\neq l}^n x_{kl} + \sum \limits_{k=1}^n \sum \limits_{l=1,k\neq l}^n x_{kl} \sum \limits_{i=1}^n b_{ikl} \sum \limits_{j=1,i\neq j}^n x_{ij} \\
  &\qquad+ \sum \limits_{i=1}^n \sum \limits_{j=1,j\neq i}^n x_{ij} \sum \limits_{l=1}^n d_{ijl} \sum \limits_{k=1,k\neq l}^n x_{kl} + \sum \limits_{k=1}^n \sum \limits_{l=1,k\neq l}^n x_{kl} \sum \limits_{j=1}^n h_{jkl} \sum \limits_{i=1,i\neq j}^n x_{ij} \\
  &=\sum \limits_{i=1}^n \sum \limits_{j=1, i\neq j}^n ( \alpha_{ij} + \beta_{ij} + \gamma_{ij} + \delta_{ij}) x_{ij}
  = \sum \limits_{i=1}^n \sum \limits_{j=1, j\neq i}^n c_{ij}x_{ij}
  = C(\tau),
\end{align*}
which concludes the proof.
\end{proof}

The result in the previous lemma can be strengthened further by excluding the conditions of the lemma for some appropriate combinations of $i,j,k$ and $\ell$.

So far we have been considering complete directed graphs.
Corresponding results can be derived easily to solve the linearization problem for QTSP on a complete undirected graph $K_n$.

An interesting related question is to explore the linearization problem of QTSP on other meaningful classes of graphs.
The answer will very much depend on the structure of the graph.
For example, if the underlying graph is a wheel, any quadratic cost matrix $Q$ is linearizable.
Note that a wheel on $n$ nodes contains $n-1$ tours and we can write a system of $n-1$ linear equations, where each equation corresponding to tour, and the variables are elements of a 'trial' linearization.
It can be shown that this system is always consistent.
We leave the question of characterizing linearizable quadratic cost matrices of specially structured graphs as a topic for further research.

\subsection*{Acknowledgement}
This project was initiated jointly with Santosh Kabadi, who passed away.
We acknowledge the significant contributions Santosh made in obtaining the results of this paper.
The work is supported by an NSERC discovery grant and an NSERC discovery accelerator supplement awarded to Abraham P.\ Punnen.

\bibliographystyle{plain}

\end{document}